\documentclass[10pt,english]{article}
\usepackage[T1]{fontenc}
\usepackage[latin9]{inputenc}
\usepackage[a4paper]{geometry}
\usepackage{color}
\usepackage{babel}
\usepackage{verbatim}
\usepackage{float}
\usepackage{amsmath}
\usepackage{amsthm}
\usepackage{amssymb}
\usepackage{esint}
\usepackage[numbers]{natbib}
\usepackage[unicode=true,pdfusetitle,
 bookmarks=true,bookmarksnumbered=false,bookmarksopen=false,
 breaklinks=false,pdfborder={0 0 0},pdfborderstyle={},backref=false,colorlinks=true]
 {hyperref}
\hypersetup{
 pdfborderstyle={},pdfborderstyle={},citecolor=blue,urlcolor=blue}
\usepackage{breakurl}

\makeatletter

\providecommand{\tabularnewline}{\\}
\floatstyle{ruled}
\newfloat{algorithm}{tbp}{loa}
\providecommand{\algorithmname}{Algorithm}
\floatname{algorithm}{\protect\algorithmname}

\newcommand{\lyxaddress}[1]{
	\par {\raggedright #1
	\vspace{1.4em}
	\noindent\par}
}
\theoremstyle{plain}
\newtheorem{thm}{\protect\theoremname}
\theoremstyle{plain}
\newtheorem{prop}[thm]{\protect\propositionname}
\theoremstyle{definition}
\newtheorem{example}[thm]{\protect\examplename}

\usepackage{babel}
\usepackage{babel}

  \providecommand{\propositionname}{Proposition}

\providecommand{\theoremname}{Theorem}

  \providecommand{\propositionname}{Proposition}
  
\providecommand{\theoremname}{Theorem}

\makeatother

\providecommand{\examplename}{Example}
\providecommand{\propositionname}{Proposition}
\providecommand{\theoremname}{Theorem}

\begin{document}

\title{Ensemble Rejection Sampling}

\author{George Deligiannidis$^{1}$, Arnaud Doucet$^{1}$\thanks{corresponding author: doucet@stats.ox.ac.uk}~~and
Sylvain Rubenthaler$^{2}$}
\maketitle

\lyxaddress{$^{1}$Department of Statistics, University of Oxford, UK.}

\lyxaddress{$^{2}$Laboratoire Jean Dieudonné, Université Nice Sophia Antipolis,
France.}
\begin{abstract}
We introduce Ensemble Rejection Sampling, a scheme for exact simulation
from the posterior distribution of the latent states of a class of
non-linear non-Gaussian state-space models. Ensemble Rejection Sampling
relies on a proposal for the high-dimensional state sequence built
using ensembles of state samples. Although this algorithm can be interpreted
as a rejection sampling scheme acting on an extended space, we show
under regularity conditions that the expected computational cost to
obtain an exact sample increases cubically with the length of the
state sequence instead of exponentially for standard rejection sampling.
We demonstrate this methodology by sampling exactly state sequences
according to the posterior distribution of a stochastic volatility
model and a non-linear autoregressive process. We also present an
application to rare event simulation.
\end{abstract}
{\small{}{}{}Keywords: }Exact simulation; Feynman\textendash Kac
models; Hidden Markov models; Rare event simulation; Rejection sampling;
State-space models.

\section{Introduction\label{section:introduction}}

Rejection sampling (RS) is a standard algorithm introduced by John
von Neumann to sample exactly from distributions only known up to
a normalizing constant. This is achieved by thinning samples from
a suitable proposal distribution; see, e.g., \citep[Chapter 2]{Devroye1986}.
However, RS is deemed inefficient to sample high-dimensional distributions
as the computational cost required to obtain an exact sample increases
typically exponentially with the dimension. This serious limitation
has partly motivated the development of Markov chain Monte Carlo (MCMC)\ and
particle methods.

We are interested here in sampling from the posterior distribution
of a sequence of latent states of length $T$ of a non-linear state-space
model. As $T$ is large in most applications of interest, RS is inefficient
in this context even when the dimension of a single state is small.
Consequently, one typically relies instead on particle smoothing schemes
\citep{GodsillDoucetWest2004} or MCMC techniques such as Embedded
Hidden Markov Model (EHMM) \citep{Neal2003,Neal2004} or particle
MCMC \citep{AndrieuDoucetHolenstein2010}. However, particle smoothing
schemes only return approximate samples from the posterior. Similarly,
as any MCMC scheme initialized out-of-equilibrium, EHMM and particle
MCMC algorithms provide biased samples after a finite number of iterations.
While many MCMC-based algorithms to sample exactly from complex probability
distributions have been developed following the introduction of coupling
from the past \citep{ProppWilson1996} - see, e.g., \citep{Kendall2005}
and \citep{Murdoch1998} -, none of the procedures developed in these
contributions appear applicable to posterior simulation in state-space
models. To our knowledge, the only exact simulation procedure available
for such models has been recently proposed in \citep{ChopinDoucetRubenthaler2019}
and relies on a combination of dominated coupling from the past and
an original version of particle MCMC based on branching processes.
We follow here an alternative approach based on RS. Contrary to the
methodology proposed in \citep{ChopinDoucetRubenthaler2019}, it is
only applicable to models where the target distribution satisfies
a backward Markovian decomposition (see, e.g., \citep{GodsillDoucetWest2004}
and \citep[Proposition 9.14][]{Douc2013}) but it is significantly
easier to analyze and implement. Additionally, this algorithm does
not suffer from the ``user-impatience'' bias problem of procedures
based on the coupling-from-the-past protocol \citep{Kendall2005}.

We show that it is indeed possible to leverage the structural properties
of state-space models to perform exact simulation using RS ideas.
This is achieved by introducing a novel RS scheme which we call Ensemble
Rejection Sampling (ERS) as it builds a proposal for the high-dimensional
state sequence based on ensembles of state samples which can be sampled
efficiently using dynamic programming techniques. This proposal is
similar to the one used in a MCMC algorithm recently introduced in
\citep[Section 3.2]{FinkeDoucetJohansen2016} which relies on EHMM
ideas. For the corresponding acceptance probability not to vanish
exponentially fast with $T$, we introduce a novel auxiliary target
distribution which admits a marginal distribution coinciding with
the posterior distribution of interest while being ``close'' to the
proposal distribution. This auxiliary target distribution differs
from the one introduced in \citep{Neal2003,Neal2004} and used in
\citep{FinkeDoucetJohansen2016,Shestopaloff2016} to establish the
validity of the MCMC algorithms based on EHMM ideas proposed therein.

Under a strong regularity assumption in the spirit of the assumptions
used in the literature to establish quantitative bounds for particle
methods \citep[Chapter 4]{DelMoral2004}, we show that the expected
computational cost to obtain one exact sample from the posterior increases
only cubically with $T$ using ERS instead of exponentially using
standard RS. We demonstrate the algorithm on two state-space models
and a rare event problem.

The rest of the paper is organized as follows. For pedagogical reasons,
we first introduce the ERS methodology in a simple `static' scenario
in Section \ref{sec:Ensemble-Rejection-Samplingstatic}. In this scenario,
ERS is of no practical interest but this allows us to introduce the
main ideas behind our construction. Section \ref{sec:Ensemble-Rejection-Samplingdyn}
presents and analyzes ERS for the more complex scenario of (dynamic)
state-space models. We present a few applications of ERS in Section
\ref{sec:examples}.

\section{Ensemble Rejection Sampling: Static Case\label{sec:Ensemble-Rejection-Samplingstatic}}

\subsection{Set-up and algorithm}

In this section, we are interested in sampling exactly from a distribution
on a measurable space $\left(\mathcal{X},\mathbf{\mathbb{X}}\right)$
admitting a density of the form
\begin{equation}
\pi\left(x\right)=\frac{\gamma\left(x\right)}{Z}\label{eq:targetstatic}
\end{equation}
w.r.t. a suitable dominating measure denoted $\mathrm{d}x$. We assume
that one can evaluate $\gamma:\mathcal{X}\mathbf{\rightarrow}\mathbb{R}^{+}$
pointwise while the normalizing constant $Z=\int_{\mathcal{X}}\gamma\left(x\right)\mathrm{d}x$
is typically not available.\ We also assume that we have access to
a proposal distribution admitting a density $q\left(x\right)$ w.r.t.
$\mathrm{d}x$ such that
\[
\sup_{x\in\mathcal{X}}\text{ }w\left(x\right)\leq\overline{w}<\infty,\text{ \ where \ }w\left(x\right):=\frac{\gamma\left(x\right)}{q\left(x\right)},
\]
the bound $\overline{w}$ being known.

When performing RS to sample from $\pi$ using $q$, the average acceptance
probability of a proposal is $p_{\text{RS}}=Z/\overline{w}$; see,
e.g., \citep[Chapter 2]{Devroye1986}. If one has access to an ensemble
of $N$ proposals $X^{1},...,X^{N}\overset{\text{i.i.d.}}{\sim}q$,
we show here that we can use this ensemble to obtain a new proposal
$X$ which will have a higher average acceptance probability. This
is achieved by computing the following normalized importance sampling
approximation of $\pi$
\begin{equation}
\widehat{\pi}\left(\cdot\right)=\sum_{i=1}^{N}W^{i}\delta_{X^{i}}\left(\cdot\right),\label{eq:ISapproxTarget}
\end{equation}
where
\begin{equation}
W^{i}=\frac{w(X^{i})}{N\widehat{Z}},\quad\widehat{Z}=\frac{1}{N}\sum_{i=1}^{N}w(X^{i}),\label{eq:Zhat}
\end{equation}
then sampling from this approximation, $X\sim\widehat{\pi}$, as in
sampling importance resampling \citep{Rubin1988,Smith1992}. Note
that $\widehat{Z}$ is an unbiased estimate of $Z$. If we denote
$\mathbf{X}:=(X_{1},\ldots,X_{N})$, the probability distribution
of the proposal $X$ has a density given by
\begin{equation}
\overline{q}\left(\cdot\right)=\mathbb{E}_{\mathbf{X}}\left[\widehat{\pi}\left(\cdot\right)\right].\label{eq:proposal}
\end{equation}
However, we cannot compute the corresponding density pointwise as
it is given by an intractable expectation. The ERS scheme described
in Algorithm \ref{alg:RES} bypasses this issue by accepting the proposal
$X$ with probability $\widehat{Z}/\overline{Z}$ where $\overline{Z}$
is an upper bound on $\widehat{Z}$ built using all the samples except
the one selected as the proposal. To avoid unnecessary complications,
we also assume here that $w\left(x\right)>0$ $q-$almost surely\footnote{If this assumption is not satisfied and $w\left(X^{1}\right)=\cdots=w\left(X^{N}\right)=0$,
then we cannot build $\widehat{\pi}$ in Step 2 and return to Step
1.}.

\begin{algorithm}[H]
\caption{\textbf{Ensemble Rejection Sampling: Static Case}\label{alg:RES}}

\begin{enumerate}
\item Sample $X^{1},...,X^{N}\overset{\text{i.i.d.}}{\sim}q$.
\item Sample $X\sim\widehat{\pi}$ given by (\ref{eq:ISapproxTarget}) and
compute $\widehat{Z}$ given by (\ref{eq:Zhat}).
\item Compute the upper bound $\overline{Z}=\widehat{Z}+\frac{1}{N}\left(\overline{w}-w(X)\right)$
on $\widehat{Z}$.
\item With probability
\begin{equation}
\frac{\widehat{Z}}{\overline{Z}},\label{eq:acceptproba}
\end{equation}
output $X$. Otherwise, return to Step 1.
\end{enumerate}
\end{algorithm}

\subsection{Proof of correctness}

The ERS scheme described in Algorithm \ref{alg:RES} satisfies the
following property.
\begin{prop}
\label{prop:ERSstatic} For any $N\geq1$, Algorithm \ref{alg:RES}
returns an exact sample from the distribution $\pi$ defined in (\ref{eq:targetstatic})
and the average acceptance probability \textup{$p_{\text{ERS}}$}
of a proposal satisfies
\[
p_{\textup{ERS}}\geq\frac{Np_{\textup{RS}}}{1+(N-1)p_{\textup{RS}}}.
\]
In particular, one has \textup{$p_{\text{ERS}}\rightarrow1$} as $N\rightarrow\infty$.
\end{prop}

\begin{proof}
For $N=1$, the result is trivial as ERS coincides with the standard
RS algorithm. To establish this result for $N\geq2$, we show that
Algorithm \ref{alg:RES} is a standard RS scheme sampling a target
distribution defined on an extended space; this extended target distribution
admitting a marginal distribution equal to $\pi$. Sampling from the
proposal in Step 2 of Algorithm \ref{alg:RES} can be rewritten as
sampling from the distribution
\begin{equation}
\overline{q}(k,\mathbf{x},x)=\left\{ {\textstyle \prod\nolimits _{i=1}^{N}}q(x^{i})\right\} \text{ }\frac{w(x^{k})}{\sum_{i=1}^{N}w(x^{i})}\text{ }\delta_{x^{k}}\left(x\right),\label{eq:extendedproposal}
\end{equation}
where $\mathbf{x}:=(x^{1},...,x^{N})$; i.e. sample $X^{1},...,X^{N}\overset{\text{i.i.d.}}{\sim}q$
then sample an index $K\in\left\{ 1,...,N\right\} $ where $\mathbb{P}(K=k)\propto w(x^{k})$
and set $X=X^{K}.$ To simplify notation, we avoid here measure-theoretic
notation for the Dirac measure. We can indeed easily check that the
distribution of $X$ under (\ref{eq:extendedproposal}) is equal to
(\ref{eq:proposal}). As this proposal is accepted with probability
(\ref{eq:acceptproba}), accepted samples are distributed according
to a distribution $\overline{\pi}(k,\mathbf{x},x)$ satisfying
\begin{align}
\overline{\pi}(k,\mathbf{x},x) & \propto\overline{q}(k,\mathbf{x},x)\frac{\widehat{Z}}{\overline{Z}}.\label{eq:targetbarstaticgen}
\end{align}
Elementary calculations show that
\begin{align}
\overline{q}(k,\mathbf{x},x)\text{ }\frac{\widehat{Z}}{Z} & =\frac{\pi\left(x^{k}\right)}{N}\delta_{x^{k}}\left(x\right){\textstyle \prod\limits _{i=1,i\neq k}^{N}}q(x^{i}):=\widetilde{\pi}\left(k,\mathbf{x},x\right).\label{eq:identitypmcmc}
\end{align}
The identity (\ref{eq:identitypmcmc}) has been used for example in
\citep[Theorem 1]{AndrieuDoucetHolenstein2010}. By using (\ref{eq:identitypmcmc}),
we can rewrite (\ref{eq:targetbarstaticgen}) as
\begin{equation}
\overline{\pi}(k,\mathbf{x},x)\propto\pi\left(x^{k}\right)\delta_{x^{k}}\left(x\right)\underset{\text{term independent of }x=x^{k}}{\underbrace{\frac{{\textstyle \prod\limits _{i=1,i\neq k}^{N}}q(x^{i})}{\overline{Z}}.}}\label{eq:targetpibarstatic}
\end{equation}
The last term on the right hand side of (\ref{eq:targetpibarstatic})
does not indeed depend on $x=x^{k}$ as $\overline{Z}=\frac{1}{N}(\overline{w}+\sum_{i=1,i\neq k}^{N}w(x^{i})$).
From (\ref{eq:targetpibarstatic}), it follows directly that the marginal
distribution of $X$ under $\overline{\pi}$ satisfies
\begin{align*}
\overline{\pi}\left(x\right) & :={\textstyle \sum\nolimits _{k=1}^{N}}\int\cdots\int_{\mathcal{X}^{N}}\overline{\pi}(k,\mathbf{x},x)\mathrm{d}\mathbf{x}\\
 & =\pi\left(x\right).
\end{align*}

Hence we have shown that Algorithm \ref{alg:RES}\ is a RS algorithm
targeting $\overline{\pi}(k,\mathbf{x},x)$ using the proposal $\overline{q}(k,\mathbf{x},x)$.
As $\overline{\pi}\left(x\right)=\pi\left(x\right)$, this returns
in particular a sample from $\pi$.

The average acceptance probability of a proposal for ERS is given
by
\begin{align*}
p_{\text{ERS}}=\mathbb{E}_{\overline{q}}\left[\frac{\widehat{Z}}{\overline{Z}}\right] & =Z\mathbb{E}_{\widetilde{\pi}}\left[\frac{1}{\overline{Z}}\right]\\
 & \geq\frac{Z}{\mathbb{E}_{\widetilde{\pi}}\left[\thinspace\overline{Z}\thinspace\right]}\\
 & =\frac{NZ}{\left(N-1\right)Z+\overline{w}}\\
 & =\frac{Np_{\text{RS}}}{(N-1)p_{\text{RS}}+1},
\end{align*}

where we have exploited, in turn, the identity (\ref{eq:identitypmcmc}),
Jensen's inequality and the fact that $X_{i}\sim q$ under $\widetilde{\pi}$
for $i\neq K$ and that $\overline{Z}$ only depends on $\mathbf{X}\setminus\{X^{K}\}$.
\end{proof}
An obvious competitor to ERS consists of simply running $N$ independent
standard RS algorithms, this has approximately the same computational
complexity and the probability of obtaining at least one sample from
$\pi$ is $1-(1-p_{\text{RS}})^{N}$. In our experiments ERS never
outperforms this simple procedure in terms of expected computational
efforts required to obtain an exact sample from $\pi$. Algorithm
\ref{alg:RES} is thus of limited practical interest. %
However, we will see in the next section that the main ideas behind
this algorithm can be extended to propose a useful exact simulation
algorithm for state-space models.

\section{Ensemble Rejection Sampling: Dynamic Case\label{sec:Ensemble-Rejection-Samplingdyn}}

\subsection{Algorithm}

Let $z_{i:j}:=(z_{i},z_{i+1},...,z_{j})$ for $i\leq j$ and $[m]:=\{1,2,...,m\}$
for any integer $m\geq1$. We now consider the case where we are interested
in sampling the posterior distribution of the latent states $X_{1:T}\in\mathcal{X}^{T}$
of a state-space model given a realization $Y_{1:T}=y_{1:T}$ of the
observations, i.e. we consider a target distribution of density
\begin{equation}
\pi(x_{1:T})=\frac{\gamma(x_{1:T})}{Z},\label{eq:targetSSM}
\end{equation}
where
\begin{equation}
\gamma(x_{1:T})=p(x_{1:T},y_{1:T})=\mu(x_{1})g(y_{1}|x_{1}){\textstyle \prod\nolimits _{t=2}^{T}}f(\left.x_{t}\right\vert x_{t-1})g(y_{t}|x_{t}),\label{eq:unnormalizedtargetSSM}
\end{equation}
$\mu$ being the initial prior density of the latent Markov process,
$f$ its homogeneous transition density and $g$ defining the conditional
density of the observations; all these densities being defined w.r.t.
to suitable dominating measures. The normalizing constant is thus
given by
\[
Z=p(y_{1:T})=\int\cdots\int_{\mathcal{X}^{T}}p(x_{1:T},y_{1:T})\mathrm{d}x_{1:T}.
\]
The algorithm detailed below is more generally applicable to a class
of Feynman-Kac models \citep{DelMoral2004}, i.e. replace $g(y_{t}|x_{t})$
in (\ref{eq:unnormalizedtargetSSM}) by a non-negative function $G_{t}(x_{t})$;
see Section \ref{sec:examples} for an application to rare event simulation.

As in Section \ref{sec:Ensemble-Rejection-Samplingstatic}, we will
sample from the target using a RS mechanism based on a proposal built
using an importance sampling approximation of $\pi$. This approximation
is obtained as follows. We draw $X_{t}^{1},...,X_{t}^{N}\overset{\text{i.i.d.}}{\sim}q_{t}$
for $t\in[T]$, these samples define a random grid in $\mathcal{X}^{T}$.
From these ensembles of samples, we can build $N^{T}$ paths $X_{1:T}^{i_{1:T}}:=(X_{1}^{i_{1}},...,X_{T}^{i_{T}})$
where $i_{1:T}\in\left[N\right]^{T}$ and each such path is marginally
distributed according to $\prod\nolimits _{t=1}^{T}q_{t}(x_{t})$.
Thus a self-normalized importance sampling approximation of $\pi$
is given by
\begin{equation}
\widehat{\pi}\left(\cdot\right)=\sum_{i_{1:T}\in\left[N\right]^{T}}W^{i_{1:T}}\delta_{X_{1:T}^{i_{1:T}}}\left(\cdot\right),\label{eq:ISapproxtargetjoint}
\end{equation}
where
\begin{equation}
W^{i_{1:T}}=\frac{w(X_{1:T}^{i_{1:T}})}{N^{T}\widehat{Z}},\quad\widehat{Z}=\frac{1}{N^{T}}\sum_{i_{1:T}\in\left[N\right]^{T}}w(X_{1:T}^{i_{1:T}}).\label{eq:ZhatHMM}
\end{equation}
Here the unnormalized importance weights are given by
\[
w(x_{1:T}^{i_{1:T}})=\frac{p(x_{1:T}^{i_{1:T}},y_{1:T})}{{\textstyle \prod\nolimits _{t=1}^{T}q_{t}(x_{t}^{i_{t}})}}=w_{1}(x{}_{1}^{i_{1}}){\textstyle \,\prod\nolimits _{t=2}^{T}}\,w_{t}(x{}_{t-1}^{i_{t-1}},x_{t}^{i_{t}}),
\]
where the `incremental' importance weights satisfy
\[
w_{1}(x_{1}):=\frac{\mu(x_{1})g(y_{1}|x_{1})}{q_{1}(x_{1})},\qquad w_{t}(x_{t-1},x_{t}):=\frac{f(\left.x_{t}\right\vert x_{t-1})g(\left.y_{t}\right\vert x_{t})}{q_{t}(x_{t})}~\textrm{for}\ t\geq2.
\]
Note that $\widehat{Z}$ is an unbiased estimator of $Z$. Even if
$\pi$ is defined on the high-dimensional space $\mathcal{X}^{T}$,
importance sampling can be expected to provide a decent approximation
to $\pi$ as $\widehat{\pi}$ relies on an exponentially large number
$N^{T}$ of (dependent) samples. Once the random grid has been sampled,
our proposal is simply obtained by sampling $X_{1:T}\sim\widehat{\pi}$
given in (\ref{eq:ISapproxtargetjoint}). If we denote $\mathbf{X}_{t}=(X_{t}^{1},...,X_{t}^{N})$
the random samples generated at time $t$, the marginal distribution
of the proposal $X_{1:T}$ is thus given by
\begin{equation}
\overline{q}\left(\cdot\right)=\mathbb{E}_{\mathbf{X}_{1:T}}\left[\widehat{\pi}\left(\cdot\right)\right].\label{eq:proposalembeddedHMM}
\end{equation}

This proposal has been introduced in an independent Metropolis\textendash Hastings
scheme in \citep[Section 3.2]{FinkeDoucetJohansen2016} and it is
a slight variation over the approach proposed in \citep{Neal2003}\footnote{
\title{In \citep{Neal2003}, a MCMC scheme is proposed to sample $\pi$.
At each iteration, $N-1$ particles $X_{t}^{1},...,X_{t}^{N-1}$ are
sampled for each $t\in[T]$ using two MCMC kernels of invariant distribution
$q_{t}$ initialized using the component $X_{t}$ of a reference path.
At equilibrium, the reference path is distributed according to $\pi$.}\maketitle
}.

We will assume from now on that
\begin{equation}
\sup_{x_{1}}w_{1}(x_{1})\leq\overline{w}_{1}<\infty,\qquad\sup_{x_{t-1},x_{t}}w_{t}(x_{t-1},x_{t})\leq\overline{w}_{t}<\infty~\textrm{for}\ t\geq2,\label{eq:importanceweightsbounded}
\end{equation}
where these supremums only have to be taken over the support of the
corresponding marginals of $\pi$ and $\overline{w}_{1},...,\overline{w}_{T}$
are known. We might have also access to some upper bounding functions
for $t\geq2$
\begin{equation}
\sup_{x_{t}}w_{t}(x_{t-1},x_{t})\leq\overline{w}_{t}^{1}(x_{t-1}),\qquad\sup_{x_{t-1}}w_{t}(x_{t-1},x_{t})\leq\overline{w}_{t}^{2}(x_{t}).\label{eq:importanceweigtspartialbounds}
\end{equation}
We can always select $\overline{w}_{t}^{1}(x_{t-1})=\overline{w}_{t}$
and $\overline{w}_{t}^{2}(x_{t})=\overline{w}_{t}$ if tighter bounding
functions are not available.

As observed in \citep{Neal2003}, a key point is that it is possible
to sample exactly from proposals of the form (\ref{eq:proposalERSdyn})
in $O(N^{2}T)$ operations despite the fact that it is a discrete
distribution defined on a space of cardinality $N^{T}$. This can
be achieved by using a HMM-type recursion detailed in Algorithm \ref{alg:HMMrecursion}
for the `embedded' $N$-state Markov chain of state space $\left\{ X_{t}^{1},...,X_{t}^{N}\right\} $
at time $t$, initial probability proportional to $w_{1}(X_{1}^{i})$,
transition probabilities proportional to $f(X_{t}^{j}|X_{t-1}^{i})$
and emission probabilities proportional to $g(y_{t}|X_{t}^{i})/q_{t}(X_{t}^{i})$
at time $t\geq2$. Indeed sampling from $\widehat{\pi}$ is equivalent
to sampling from the posterior distribution of this embedded HMM defined
as
\[
\widetilde{p}(x_{1:T}|y_{1:T}):=\frac{w(x_{1:T})}{N^{T}\widehat{Z}}.
\]
It is well-known that one can sample from this discrete distribution
using a forward filtering-backward sampling algorithm exploiting the
decomposition\footnote{This decomposition has also been used to sample exactly from the posterior
distribution of linear Gaussian state-space models by leveraging Kalman
recursions \citep{CarterKohn1994}.}
\begin{equation}
\widetilde{p}(x_{1:T}|y_{1:T})=\widetilde{p}(x_{T}|y_{1:T})\prod_{t=1}^{T-1}\widetilde{p}(x_{t}|y_{1:t},x_{t+1}),\label{eq:forwardbackward}
\end{equation}
where
\begin{equation}
\widetilde{p}(x_{t}|y_{1:t},x_{t+1}):=\frac{f(x_{t+1}|x_{t})\widetilde{p}(x_{t}|y_{1:t})}{\sum_{i=1}^{N}f(x_{t+1}|X_{t}^{i})\widetilde{p}(X_{t}^{i}|y_{1:t})}.\label{eq:backwardfilter}
\end{equation}
As a byproduct of the HMM forward recursion for the embedded chain,
we also compute $\widehat{Z}$ in $O(N^{2}T)$ operations; see Algorithm
\ref{alg:HMMrecursion} for details.

As in Section \ref{sec:Ensemble-Rejection-Samplingstatic}, ERS requires
being able to compute an upper bound $\overline{Z}$ on $\widehat{Z}$
built using all the samples $\mathbf{X}_{1:T}$ except the ones selected
as the proposal $X_{1:T}$. This proposal is of the form $X_{1:T}=(X_{1}^{K_{1}},...,X_{T}^{K_{T}})$
for $(K_{1},...,K_{T})\in\left[N\right]^{T}$ such that $\mathbb{\mathrm{Pr}}\left(\left(K_{1},...,K_{T}\right)=\left(k_{1},...,k_{T}\right)\right)=\widetilde{p}(x_{1:T}^{k_{1:T}}|y_{1:T})$.
From (\ref{eq:ZhatHMM}), $\widehat{Z}$ is an average of $N^{T}$
products of $T$ terms of the form $w_{1}(X{}_{1}^{i_{1}}){\textstyle \prod\nolimits _{t=2}^{T}}w_{t}(X{}_{t-1}^{i_{t-1}},X_{t}^{i_{t}})$.
We obtain $\overline{Z}$ by upper bounding any term $w_{1}(X{}_{1}^{K_{1}})$
by $\overline{w}_{1}$, $w_{t}(X{}_{t-1}^{i_{t-1}},X_{t}^{K_{t}})$
by $\overline{w}_{t}^{1}(X_{t-1}^{i_{t-1}})$ for $i_{t-1}\neq K{}_{t-1}$,
$w_{t}(X{}_{t-1}^{K_{t-1}},X_{t}^{i_{t}})$ by $\overline{w}_{t}^{2}(X_{t}^{i_{t}})$
for $i_{t}\neq K{}_{t}$ and $w_{t}(X{}_{t-1}^{K_{t-1}},X_{t}^{K_{t}})$
by $\overline{w}_{t}$ for $t=2,...,T$. This bound can be computed
in $O(N^{2}T)$ operations using a simple modification of the HMM
recursion presented in Algorithm \ref{alg:HMMrecursion}; this is
detailed in Algorithm \ref{alg:boundingHMMrecursion}. Algorithm \ref{alg:RES-Dyn}
summarizes the ERS scheme. For $T=1$, this algorithm corresponds
to Algorithm \ref{alg:RES} applied to $\gamma(x)=\mu(x)g(y_{1}|x)$
using $q(x)=q_{1}(x)$. We also assume that $w(x_{1:T})>0$ ${\textstyle \prod\nolimits _{t=1}^{T}q_{t}-}$almost
surely to simplify presentation.

\begin{algorithm}[H]
\caption{\textbf{Ensemble Rejection Sampling: Dynamic Case}\label{alg:RES-Dyn}}

\begin{enumerate}
\item For $t\in[T]$, sample $X_{t}^{1},...,X_{t}^{N}\overset{\text{i.i.d.}}{\sim}q_{t}$.
\item Sample $X_{1:T}\sim\widehat{\pi}$ given by (\ref{eq:ISapproxtargetjoint})
and compute $\widehat{Z}$ given by (\ref{eq:ZhatHMM}) using Algorithm
\ref{alg:HMMrecursion}.
\item Compute an upper bound $\overline{Z}$ on $\widehat{Z}$ using Algorithm
\ref{alg:boundingHMMrecursion}.
\item With probability
\begin{equation}
\frac{\widehat{Z}}{\overline{Z}},\label{eq:acceptprobaERSdyn}
\end{equation}
output $X_{1:T}$. Otherwise, return to Step 1.
\end{enumerate}
\end{algorithm}

\begin{algorithm}[H]
\caption{\textbf{Hidden Markov Model recursion}\label{alg:HMMrecursion}}

\begin{enumerate}
\item At time $t=1$
\begin{enumerate}
\item For $i\in[N]$, set $\widetilde{p}(X_{1}^{i},y_{1})=w_{1}(X_{1}^{i})$.
\item Compute $\widetilde{p}(y_{1})=\sum_{i=1}^{N}w_{1}(X_{1}^{i})$.
\item For $i\in[N]$, compute $\widetilde{p}(X_{1}^{i}|y_{1})=\widetilde{p}(X_{1}^{i},y_{1})/\widetilde{p}(y_{1})$.
\end{enumerate}
\item For $t=2,...,T$
\begin{enumerate}
\item For $i\in[N]$, compute $\widetilde{p}(X_{t}^{i},y_{t}|y_{1:t-1})=\sum_{j=1}^{N}\widetilde{p}(X_{t-1}^{j}|y_{1:t-1})w_{t}(X_{t-1}^{j},X_{t}^{i}$).
\item Compute $\widetilde{p}(y_{t}|y_{1:t-1})=\sum_{i=1}^{N}\widetilde{p}(X_{t}^{i},y_{t}|y_{1:t-1})$
and $\widetilde{p}(y_{1:t})=\widetilde{p}(y_{1:t-1})\widetilde{p}(y_{t}|y_{1:t-1})$.
\item For $i\in[N]$, compute $\widetilde{p}(X_{t}^{i}|y_{1:t})=\widetilde{p}(X_{t}^{i},y_{t}|y_{1:t-1})/\widetilde{p}(y_{t}|y_{1:t-1})$.
\end{enumerate}
\item Sample $X_{T}^{K_{T}}\sim\widetilde{p}\left(\cdot|y_{1:T}\right)$.
\item For $t=T-1,...,1$, sample $X_{t}^{K_{t}}\sim\widetilde{p}(\cdot|y_{1:t},X_{t+1}^{K_{t+1}})$
(see equation (\ref{eq:backwardfilter})).
\item Output $X_{1:T}:=X_{1:T}^{K_{1:T}},$ $K_{1:T}$ and $\widehat{Z}=\widetilde{p}(y_{1:T})/N^{T}.$
\end{enumerate}
\end{algorithm}

\begin{algorithm}[H]
\caption{\textbf{Bounding Hidden Markov Model recursion}\label{alg:boundingHMMrecursion}}

\begin{enumerate}
\item At time $t=1$
\begin{enumerate}
\item Set $\overline{p}(X_{1}^{K_{1}},y_{1})=\overline{w}_{1}$ and for
$i\in[N]\setminus\{K_{1}\}$, set $\overline{p}(X_{1}^{i},y_{1})=w_{1}(X_{1}^{i})$.
\item Compute $\overline{p}(y_{1})=\sum_{i=1}^{N}\overline{p}(X_{1}^{i},y_{1})$.
\item For $i\in[N]$, compute $\overline{p}(X_{1}^{i}|y_{1})=\overline{p}(X_{1}^{i},y_{1})/\overline{p}(y_{1})$.
\end{enumerate}
\item For $t=2,...,T$
\begin{enumerate}
\item For $i\in[N]\setminus\{K_{t}\}$, compute
\[
\overline{p}(X_{t}^{i},y_{t}|y_{1:t-1})=\sum_{j\neq K_{t-1}}\overline{p}(X_{t-1}^{j}|y_{1:t-1})w_{t}(X_{t-1}^{j},X_{t}^{i})+\overline{p}(X_{t-1}^{K_{t-1}}|y_{1:t-1})\overline{w}_{t}^{2}(X_{t}^{i}).
\]
\item For $i=K_{t}$, compute
\[
\overline{p}(X_{t}^{K_{t}},y_{t}|y_{1:t-1})=\sum_{j\neq K_{t-1}}\overline{p}(X_{t-1}^{j}|y_{1:t-1})\overline{w}_{t}^{1}(X_{t-1}^{j})+\overline{p}(X_{t-1}^{K_{t-1}}|y_{1:t-1})~\overline{w}_{t}.
\]
\item Compute $\overline{p}(y_{t}|y_{1:t-1})=\sum_{i=1}^{N}\overline{p}(X_{t}^{i},y_{t}|y_{1:t-1})$
and $\overline{p}(y_{1:t})=\overline{p}(y_{t}|y_{1:t-1})\overline{p}(y_{1:t-1})$.
\item For $i\in[N]$, compute $\overline{p}(X_{t}^{i}|y_{1:t})=\overline{p}(X_{t}^{i},y_{t}|y_{1:t-1})/\overline{p}(y_{t}|y_{1:t-1})$.
\end{enumerate}
\item Output $\overline{Z}=\overline{p}(y_{1:T})/N^{T}.$
\end{enumerate}
\end{algorithm}

\subsection{Theoretical results}

We establish here the validity of the ERS scheme described in Algorithm
\ref{alg:RES-Dyn}.
\begin{prop}
\label{prop:ERSdynisvalid}For any $N\geq1$, Algorithm \ref{alg:RES-Dyn}
returns an exact sample from the distribution $\pi$ defined in (\ref{eq:targetSSM})
and its corresponding average acceptance probability satisfies\textup{
$p_{\text{ERS}}\rightarrow1$} as $N\rightarrow\infty$.
\end{prop}

\begin{proof}
For $N=1$, ERS is a standard RS scheme so it returns samples from
$\pi$. The first part of the proof is essentially identical to the
first part of the proof of Proposition \ref{prop:ERSstatic}. Sampling
from the proposal in Step 2 of Algorithm \ref{alg:RES-Dyn} can be
rewritten as sampling from the distribution
\begin{equation}
\overline{q}(k_{1:T},\mathbf{x}_{1:T},x_{1:T})=\left\{ {\textstyle \prod\nolimits _{t=1}^{T}}{\textstyle \prod\nolimits _{i=1}^{N}}q_{t}(x_{t}^{i})\right\} \frac{w(x_{1:T}^{k_{1:T}})}{\sum_{i_{1:T}\in\left[N\right]^{T}}w(x_{1:T}^{i_{1:T}})}\delta_{x_{1:T}^{k_{1:T}}}(x_{1:T}),\label{eq:proposalERSdyn}
\end{equation}
where $\mathbf{x}_{t}:=(x_{t}^{1},...,x_{t}^{N}).$ Indeed, it is
easy to check that the marginal distribution of $X_{1:T}$ under $\overline{q}$
satisfies (\ref{eq:proposalembeddedHMM}).

The distribution of the samples accepted by Algorithm \ref{alg:RES-Dyn}
is thus given by
\begin{align}
\overline{\pi}(k_{1:T},\mathbf{x}_{1:T},x_{1:T}) & \propto\overline{q}(k_{1:T},\mathbf{x}_{1:T},x_{1:T})\frac{\widehat{Z}}{\overline{Z}},\label{eq:targetextendedERSdyn}
\end{align}
where we recall that $\widehat{Z}$ is a function of $\mathbf{x}_{1:T}$
while, by construction, $\overline{Z}$ is only a function of $\mathbf{x}_{1:T}\setminus\{x_{1:T}\}.$
One can verify that
\begin{align}
\overline{q}(k_{1:T},\mathbf{x}_{1:T},x_{1:T})\frac{\widehat{Z}}{Z} & =\frac{\pi(x_{1:T}^{k_{1:T}})}{N^{T}}\delta_{x_{1:T}^{k_{1:T}}}(x_{1:T}){\textstyle \prod\nolimits _{t=1}^{T}}{\textstyle \prod\nolimits _{i=1,i\neq k_{t}}^{N}}q_{t}(x_{t}^{i})\nonumber \\
 & :=\widetilde{\pi}\left(k_{1:T},\mathbf{x}_{1:T},x_{1:T}\right).\label{eq:identityEHMM}
\end{align}
The probability distribution $\widetilde{\pi}(k_{1:T},\mathbf{x}_{1:T},x_{1:T})$
was introduced in \citep{Neal2003} and the identity relating $\widetilde{\pi}$
to $\overline{q},\widehat{Z},Z$ has been established in \citep[Section 3.2]{FinkeDoucetJohansen2016}.
By using (\ref{eq:identityEHMM}), we can rewrite (\ref{eq:targetextendedERSdyn})
as
\begin{equation}
\overline{\pi}(k_{1:T},\mathbf{x}_{1:T},x_{1:T})\propto\pi(x_{1:T}^{k_{1:T}})\delta_{x_{1:T}^{k_{1:T}}}(x_{1:T})\underset{\text{term independent of \ensuremath{x_{1:T}=x_{1:T}^{k_{1:T}}} }}{\underbrace{\frac{{\textstyle \prod\nolimits _{t=1}^{T}}{\textstyle \prod\nolimits _{i=1,i\neq k_{t}}^{N}}q_{t}\left(x_{t}^{i}\right)}{\overline{Z}}}}.\label{eq:targetpibarstatic-1}
\end{equation}
From (\ref{eq:targetpibarstatic-1}), it follows directly that
\begin{align*}
\overline{\pi}(x_{1:T}) & :={\textstyle \sum_{k_{1:T}\in\left[N\right]^{T}}}\int\cdots\int_{\mathcal{X}^{TN}}\overline{\pi}(k_{1:T},\mathbf{x}_{1:T},x_{1:T})\mathrm{d}\mathbf{x}_{1:T}\\
 & =\pi(x_{1:T}).
\end{align*}
Hence we have shown that Algorithm \ref{alg:RES-Dyn} is a RS algorithm
targeting $\overline{\pi}(k_{1:T},\mathbf{x}_{1:T},x_{1:T})$ using
the proposal $\overline{q}(k_{1:T},\mathbf{x}_{1:T},x_{1:T})$. As
$\overline{\pi}(x_{1:T})=\pi(x_{1:T})$, this returns in particular
a sample from $\pi$.

Using arguments similar to the ones used in the proof of Proposition
\ref{prop:ERSstatic}, we can show that the average acceptance probability
$p_{\text{ERS}}$ of a proposal for ERS satisfies
\begin{align*}
p_{\text{ERS}}=\mathbb{E}_{\overline{q}}\left[\frac{\widehat{Z}}{\overline{Z}}\right] & \geq\frac{Z}{\mathbb{E}_{\widetilde{\pi}}\left[\thinspace\overline{Z}\thinspace\right]},
\end{align*}
where $\widetilde{\pi}$ is defined in (\ref{eq:identityEHMM}). By
bounding any term of the sum appearing in $\overline{Z}$ involving
at least one index such that $i_{t}=K_{t}$ for $t\in[T]$ by $\prod_{t=1}^{T}\overline{w}_{t}$,
we obtain the following upper bound
\begin{align}
\mathbb{E}_{\widetilde{\pi}}\left[\thinspace\overline{Z}\thinspace\right]\leq & \left(1-\frac{1}{N}\right)^{T}Z+\left(1-\left(1-\frac{1}{N}\right)^{T}\right)\prod_{t=1}^{T}\overline{w}_{t}\label{eq:crudeboundonZhat}
\end{align}
as $X_{t}^{i_{t}}\sim q_{t}$ for $t\in[T]$ and $i\neq K_{t}$ under
$\widetilde{\pi}$. Hence the r.h.s. of (\ref{eq:crudeboundonZhat})
converges to $Z$ as $N\rightarrow\infty$ and thus $p_{\text{ERS}}\rightarrow1$
as $N\rightarrow\infty$.
\end{proof}
Proposition \ref{prop:ERSdynisvalid} does not guarantee that at fixed
computational efforts the proposed algorithm ERS is competitive compared
to a standard RS algorithm using the proposal $\prod_{t=1}^{T}q_{t}(x_{t})$.
For $T=1$, we have seen in Section \ref{sec:Ensemble-Rejection-Samplingstatic}
that this is not the case. However the following toy example suggests
that ERS exhibits desirable properties for large $T$ if we scale
$N$ with $T$ appropriately.
\begin{example}
Consider a scenario where $y_{1:T}=(y,...,y)$,
\[
\gamma(x_{1:T})=p(x_{1:T},y_{1:T})={\textstyle \prod\nolimits _{t=1}^{T}}\mu(x_{t})g(y|x_{t}),
\]
so
\[
Z=\mathcal{Z}{}^{T},\text{ \ where \ }\mathcal{Z}=\int\mu(x)g(y|x)\mathrm{d}x.
\]
Obviously the corresponding target distribution $\pi$ factorizes
in $T$ independent terms in this case. However, for the sake of illustration,
we ignore this fact and apply directly ERS to $\pi$. Standard RS
using a proposal $\prod_{t=1}^{T}q_{1}(x_{t})$ has an average acceptance
probability $p_{\text{RS}}=p_{\text{A}}^{T}$, where $p_{\text{A}}=\mathcal{Z}/\overline{w}_{1}$.
As long as $q_{1}(x)$ is not equal to $\mu(x)g(y|x)/\mathcal{Z}$
almost everywhere, then $p_{\text{A}}<1$ and so $p_{\text{RS}}$
decreases exponentially fast with $T$. In the same scenario, simple
calculations shows that ERS has an average acceptance probability
satisfying
\begin{align}
p_{\text{ERS}}\geq & \frac{\left(N\mathcal{Z}\right){}^{T}}{(\left(N-1\right)\mathcal{Z}+\overline{w}_{1})^{T}}\nonumber \\
= & \frac{1}{(1+\frac{1}{N}(p_{\text{A}}^{-1}-1))^{T}}\label{eq:inequalityprobaacceptance}
\end{align}
If we select $N=\left\lceil \beta T\right\rceil $ where $\beta>0$,
then the r.h.s. of (\ref{eq:inequalityprobaacceptance}) converges
to $\exp\left\{ \beta^{-1}(1-p_{\text{A}}^{-1})\right\} $ as $T\rightarrow\infty$;
i.e. we can control $p_{\text{ERS}}$ by only increasing $N$ linearly
with $T$.

We now establish a similar result for the more realistic scenario
where the target distribution does not factorize. Our result relies
on a strong regularity condition in the spirit of the assumptions
commonly used in the particle filtering literature to establish quantitative
bounds \citep[Chapter 4]{DelMoral2004}. It will typically hold for
a compact state-space $\mathcal{X}$ or when the support of the target
is given by a compact subspace of $\mathcal{X}^{T}$.
\end{example}

\begin{prop}
\label{prop:Cubicscaling}Assume that there exist $\underline{w},\overline{w}>0$
such that $\underline{w}<w_{1}(x)<\overline{w}$ and $\underline{w}<w_{t}(x,x')<\overline{w}$
for all $x,x'\in\mathcal{X}\times\mathcal{X}$ and $t\geq2$ and let
$\delta=(\overline{w}/\underline{w})^{2}$, then the average acceptance
probability \textup{$p_{\text{ERS}}$} of Algorithm \ref{alg:RES-Dyn}
satisfies
\[
p_{\textup{ERS}}\geq\frac{1}{\left(1+\frac{\delta-1}{N}\right)^{T}}.
\]
\textup{\emph{In particular, for $N=\left\lceil \beta T\right\rceil $
where $\beta>0$, we have}}\textup{
\[
\liminf_{T\rightarrow\infty}p_{\textup{ERS}}\geq\exp\{\beta^{-1}(1-\delta)\}.
\]
}
\end{prop}

\begin{proof}
To lower bound $p_{\text{ERS}}$, we use again the inequality
\begin{align*}
p_{\text{ERS}}\geq & \frac{Z}{\mathbb{E}_{\widetilde{\pi}}\left[\thinspace\overline{Z}\thinspace\right]},
\end{align*}
where $\mathbb{\widetilde{\pi}}$ is defined in (\ref{eq:identityEHMM})
and then we exploit the assumption on the incremental weights to compute
an upper bound on $\mathbb{E}_{\widetilde{\pi}}\left[\thinspace\overline{Z}\thinspace\right]$
that is tighter than the bound (\ref{eq:crudeboundonZhat}) used in
the proof of Proposition \ref{prop:ERSdynisvalid}.

The term $\overline{Z}$ is given by a sum over the indices $i_{1:T}\in\left[N\right]^{T}$.
Of the total $N^{T}$ terms, $(N-1)^{T}$ of these terms are such
that such $i_{1}\neq K_{1},i_{2}\neq K_{2},...,i_{T}\neq K_{T}$ and
each such term has expectation $Z$ under $\widetilde{\pi}$. Among
the remaining $N^{T}-\left(N-1\right)^{T}$ terms, $T$$(N-1)^{T-1}$
have exactly one index $t_{1}\in[T]$ such that $i_{t_{1}}=K_{t_{1}}$
and more generally $\binom{{T}}{n}$$(N-1)^{T-n}$ terms have exactly
$n$ distinct indices $t_{1},t_{2},..,t_{n}\in[T]^{n}$ such that
$i_{t_{m}}=K_{t_{m}}$ for $m\in[n]$. For each such term, we have
to bound the corresponding product of the $T$ incremental weights.
Each index $i_{t_{m}}=K_{t_{m}}$ can impact at most two terms in
the product, e.g., $w_{t_{m}}^{1}(x_{t_{m}-1}^{i_{t_{m}-1}})w_{t_{m}+1}^{2}(x_{t_{m+1}}^{i_{t_{m}+1}})$
if one considers a sequence $i_{1:T}$ such that $i_{t_{m}-1}\neq K_{t_{m}-1}$
and $i_{t_{m}+1}\neq K_{t_{m}+1}$. However, the assumption on the
incremental weights yields
\begin{align*}
 & \overline{w}_{t_{m}}^{1}(x_{t_{m}-1}^{i_{t_{m}-1}})\overline{w}_{t_{m}+1}^{2}(x_{t_{m+1}}^{i_{t_{m}+1}})\\
= & \frac{\overline{w}_{t_{m}}^{1}(x_{t_{m}-1}^{i_{t_{m}-1}})\overline{w}_{t_{m}+1}^{2}(x_{t_{m+1}}^{i_{t_{m}+1}})}{w_{t_{m}}(x_{t_{m}-1}^{i_{t_{m}-1}},x_{t_{m}}^{i_{t_{m}}})w_{t_{m}+1}(x_{t_{m}}^{i_{t_{m}}},x_{t_{m+1}}^{i_{t_{m}+1}})}w_{t_{m}}(x_{t_{m}-1}^{i_{t_{m}-1}},x_{t_{m}}^{i_{t_{m}}})w_{t_{m}+1}(x_{t_{m}}^{i_{t_{m}}},x_{t_{m+1}}^{i_{t_{m}+1}})\\
\leq & \left(\frac{\overline{w}}{\underline{w}}\right)^{2}w_{t_{m}}(x_{t_{m}-1}^{i_{t_{m}-1}},x_{t_{m}}^{i_{t_{m}}})w_{t_{m}+1}(x_{t_{m}}^{i_{t_{m}}},x_{t_{m+1}}^{i_{t_{m}+1}})\\
= & \delta w_{t_{m}}(x_{t_{m}-1}^{i_{t_{m}-1}},x_{t_{m}}^{i_{t_{m}}})w_{t_{m}+1}(x_{t_{m}}^{i_{t_{m}}},x_{t_{m+1}}^{i_{t_{m}+1}}).
\end{align*}
Hence, we have
\begin{align*}
\mathbb{E}_{\widetilde{\pi}}\left[\thinspace\overline{Z}\thinspace\right]\leq & \frac{Z}{N^{T}}\sum_{n=0}^{T}\binom{{T}}{n}(N-1)^{T-n}\delta^{n}\\
= & Z\sum_{n=0}^{T}\binom{{T}}{n}\left(1-\frac{1}{N}\right)^{T-n}\left(\frac{\delta}{N}\right)^{n}\\
= & Z\left(1+\frac{\delta-1}{N}\right)^{T}.
\end{align*}
The result now follows directly.
\end{proof}
Proposition \ref{prop:Cubicscaling} ensures that the average acceptance
probability $p_{\text{ERS}}$ does not vanish as long as $N$ increases
linearly with $T$. As the computational complexity to sample from
the proposal distribution is $O(N^{2}T)$, this shows that the expected
computational cost to obtain an exact sample is of order $O(T^{3})$.

\subsection{Settings and Extensions}

We discuss here how to select the proposal distributions and propose
various extensions of Algorithm \ref{alg:RES-Dyn}.

\subsubsection{Selection of the proposal distributions}

The need for bounded importance weights, see (\ref{eq:importanceweightsbounded}),
limits the range of applicability of ERS but can be satisfied for
a useful class of non-linear state-space models. For example, consider
a model such that
\[
\sup_{x}\mu(x)\leq\overline{\mu}<\infty,\quad\sup_{x,x'}f(\left.x'\right\vert x)\leq\overline{f}<\infty,\quad g(y_{t}):=\int_{\mathfrak{\mathcal{X}}}g(\left.y_{t}\right\vert x)\mathrm{d}x<\infty,
\]
for all $t$. If we use the proposals given by
\[
q_{t}(x_{t})\propto g(\left.y_{t}\right\vert x_{t}),
\]
ERS applies with $w_{1}(x_{1})=\mu(x_{1}),\overline{w}_{1}=\overline{\mu}$
and $w_{t}(x_{t-1},x_{t})=f(x_{t}|x_{t-1}),\overline{w}_{t}=\overline{f}$
for $t\geq2$\footnote{The target is not modified by replacing $g(\left.y_{t}\right\vert x)$
by $g(\left.y_{t}\right\vert x)/g(y_{t})$ so $g(y_{t})$ does not
appear in the bounds $\overline{w}_{t}$.}. Such proposals have been used in the particle filtering literature
in \citep{LinZhangChen2005}.

We might also have a model such that
\[
f(\left.x'\right\vert x)\leq c\nu\left(x'\right)
\]
where $\nu$ is a probability density function. In this case, it is
possible to use a proposal of the form
\[
q_{t}(x_{t})\propto\nu\left(x_{t}\right)g(\left.y_{t}\right\vert x_{t})
\]
for $t\geq2$ and $q_{1}(x_{1})\propto\mu(x_{1})g(y_{1}|x_{1})$.
In this case, ERS applies with $w_{1}(x_{1})=\overline{w}_{1}=1$
and $w_{t}(x_{t-1},x_{t})=f(x_{t}|x_{t-1})/\nu\left(x_{t}\right),\overline{w}_{t}=c$
for $t\geq2$.

\subsubsection{Intractable weights}

Consider a scenario where the incremental weights are not tractable
as $\mu\left(x\right),f(x'|x)$ or $g(y|x)$ cannot be evaluated pointwise.
If one has access to a non-negative unbiased estimator of the intractable
quantity that is upper bounded almost surely by a known finite constant,
then ERS can be applied directly without any modification. The correctness
of this procedure follows from a standard auxiliary variable construction.
This shows for example that ERS can be applied to simulate exactly
from the posterior distribution of the class of diffusions considered
in \citep{Fearnhead2008} whose state is observed at discrete times
in some additive Gaussian noise.

\subsubsection{Non-integrability of $g(\left.y\right\vert x)$ w.r.t. $x$ and missing
data\label{subsec:Non-integrability}}

Consider a scenario where $\int g(\left.y_{t}\right\vert x)\mathrm{d}x<\infty$
but $\int g(\left.y_{t+1}\right\vert x)\mathrm{d}x=\infty$. We can
thus use a proposal $q\left(x|y\right)\propto g(y|x)$ at time $t$
but not at time $t+1$. To bypass this problem, we can consider instead
using an incremental weight at time $t$ of the form
\[
w_{t}(x_{t-1},x_{t}):=\frac{f(\left.x_{t}\right\vert x_{t-1})g(\left.y_{t}\right\vert x_{t})p(y_{t+1}|x_{t})}{q_{t}(x_{t})},
\]
where $p(y_{t+1}|x_{t})=\int f(\left.x_{t+1}\right\vert x_{t})g(\left.y_{t+1}\right\vert x_{t+1})\mathrm{d}x_{t+1}$.
A non-negative estimate of this quantity can be obtained if it cannot
be calculated in closed-form; e.g., sample $X_{t+1}\sim f(\left.\cdot\right\vert x_{t})$
and return $g(\left.y_{t+1}\right\vert X_{t+1})$. This can be directly
extended to scenarios where $\int g(\left.y_{k}\right\vert x)\mathrm{d}x=\infty$
for $k=t+1,...,t+\triangle$ for $\triangle\geq1$. Hence we can also
use ERS in this scenario to sample the posterior distribution of the
states $x_{t}$ such that $\int g(\left.y_{t}\right\vert x)\mathrm{d}x<\infty$
and then sampling the remaining states according to their full conditionals
of the form $p(x_{t+1:t+\triangle}|x_{t},x_{t+\triangle+1},y_{t+1:t+\triangle})$
using standard RS. Obviously, this method becomes inefficient if $\triangle$
is large.

Similarly, we might be in a scenario where $x_{t}=(z_{t},...,z_{t-\triangle+1})$
and, slightly abusing notation, $g(y_{t}|x_{t})=g(y_{t}|z_{t})$ with
$\int g(y_{t}|z_{t})\mathrm{d}z{}_{t}<\infty$. In this case, we can
also pull observations together by blocks of length $\triangle$ to
sample the latent state $z_{t}$ through a distribution proportional
to the product of $g(y_{t}|z_{t})$ over a block.

\subsubsection{Other extensions}

There are many other possible extensions of interest. We only briefly
mention two of them here. For example, we can make the number of particles
depends on the time index to address scenarios where one has outliers.
We can also extend the ERS procedures to sample exactly from the posterior
distribution of the continuous latent states of a tree-structure directed
graphical model as a generalized version of the key forward-backward
recursion is available in these scenarios; see, e.g., \citep{CowellDawidLauritzenSpiegelhalter2006,Wilkinson2002}.

\section{Examples\label{sec:examples}}

All the simulations have been performed in Matlab on a standard desktop
PC. The code and data will be soon made publicly available.

\subsection{Conditioned random walks}

To demonstrate our methodology on a rare event problem, we consider
here the problem of simulating conditioned random walks discussed
in \citep{Adorisio2018,DelMoralDoucet2004} where $\mathcal{X}=\mathbb{R}$,
$\mu\left(x\right)=\mathcal{U}(x;\mathcal{S})$ is the uniform distribution
on $\mathcal{S}$, $f(x'|x)=\mathcal{N}(x';\psi\left(x\right),\sigma^{2})$
and $G_{t}(x)=1_{\mathcal{S}}\left(x\right)$ for some bounded set
$\mathcal{S}\subset\mathbf{\mathbb{X}}.$ This models the evolution
of a particle in an absorbing medium, the particle being absorbed
whenever it steps outside $\mathcal{S}$. Here the non-negative function
$G_{t}(x)$ replaces $g(y_{t}|x_{t})$ in equation (\ref{eq:unnormalizedtargetSSM}).
In this case, the distribution $\pi$ corresponds to the distribution
of the paths of the particle conditional upon not having been absorbed
by time $T$. A method to sample exactly from this distribution has
been proposed in \citep{Adorisio2018} but it is only applicable to
very specific dynamics and sets $\mathcal{S}$. ERS is particularly
well-suited to such problems. To implement ERS, we use $q_{t}(x)=\mathcal{U}(x;\mathfrak{\mathcal{S}})$
for all $t\geq1$ so that $\overline{w}_{1}=1$ and $\overline{w}_{t}=1/\sqrt{2\pi\sigma^{2}}$
for $t\geq2$. The assumptions of Proposition \ref{prop:Cubicscaling}
are satisfied. In our simulations, we consider $\mathcal{S}=[0,1]$,
$\psi\left(x\right)=x$ and $\sigma=0.2$.

\begin{table}
\begin{centering}
\par\end{centering}
\begin{centering}
\caption{\label{tab:Estimated-average-acceptance_ConditionedRW}Estimated average
acceptance probability $p_{\text{ERS}}$ in \% computed using 500
samples as a function of $T$ and $N$}
\par\end{centering}
\centering{}%
\begin{tabular}{|c|c|c|c|}
\hline
$\widehat{p}_{\text{ERS}}$  & $N=T$ & $N=2T$ & $N=5T$\tabularnewline
\hline
$T=100$ & 3.19 & 17.29 & 49.00\tabularnewline
\hline
$T=250$ & 2.91 & 16.92 & 47.75\tabularnewline
\hline
$T=500$ & 2.82 & 16.64 & 48.50\tabularnewline
\hline
\end{tabular}
\end{table}

We consider different values for $T$ as well as corresponding values
for $N$. We estimate $p_{\text{ERS}}=\mathbb{E}_{\overline{q}}[\widehat{Z}/\overline{Z}]$
by Monte Carlo using 500 samples from $\overline{q}.$ In this scenario
and other examples considered here, the relative variance of the corresponding
estimator is small as $\widehat{Z}$ and $\overline{Z}$ are strongly
positively correlated by construction. The results are summarized
in Table \ref{tab:Estimated-average-acceptance_ConditionedRW}. For
this time-homogeneous model, the average acceptance probability is,
as expected, fairly stable across $T$ for a fixed ratio $N/T$.

\subsection{Non-linear autoregressive process}

Consider the following non-linear autoregressive model where $\mathcal{X}=\mathbb{R}$,
$\mu\left(x\right)=\mathcal{N}(x;0,1)$, $f(x'|x)=\mathcal{N}(x';\phi\tanh\left(x\right),\sigma_{v}^{2})$
and $g(y|x)=\mathcal{N}(y;x,\sigma_{w}^{2})$ with $\phi=0.9$, $\sigma_{v}=0.3$
and $\sigma_{w}=0.1$. We select $q_{t}(x_{t}|y_{t})\propto g\left(y_{t}|x_{t}\right)$
and the corresponding bounds are $\overline{w}_{1}=1/\sqrt{2\pi}$
and $\overline{w}_{t}=1/\sqrt{2\pi\sigma_{v}^{2}}$ for $t\geq2$.
For a given realization of $T=500$ data points, the estimated average
acceptance probability $p_{\text{ERS}}$ using 500 samples is $p_{\text{ERS}}$
is $0.79\%$ for $N=500$, $7.34\%$ for $N=1000$ and $24.34\%$
using $N=2000$.

\subsection{Stochastic volatility model}

Consider the following univariate stochastic volatility model where
$\mathcal{X}=\mathbb{R}$, $\mu\left(x\right)=\mathcal{N}(x;0,\sigma^{2}/(1-\phi^{2}))$,
$f(x'|x)=\mathcal{N}(x';\phi x,\sigma^{2})$ and $g(y|x)=\mathcal{N}(y;0,\beta\exp\left(x/2\right))$;
see, e.g., \citep{PittShephard1999}. We have $\log Y_{t}^{2}=X_{t}+\log\beta^{2}+W_{t}$
where $\exp\left(W_{t}\right)\sim\chi^{2}(1)$. This suggests using
for $q_{t}(x_{t}|y_{t})$ the distribution obtained by sampling $X_{t}$
using $\log y_{t}^{2}-\log\beta^{2}-W_{t}$ where $\exp\left(W_{t}\right)\sim\chi^{2}(1).$
The corresponding bounds are given by $\overline{w}_{1}=\sqrt{(1-\phi^{2})/2\pi}$
and $\overline{w}_{t}=1/\sqrt{2\pi\sigma^{2}}$ for $t\geq2$. We
apply ERS to the S\&P 500 index daily data using parameters $\phi=0.95$,
$\beta=0.7$ and $\sigma=0.3$ as in \citep{Poyiadjis2011} for $T=200$
data points corresponding to the period 09/08/1990 to 24/05/1991.
For $N=6000$, the estimated average acceptance probability $p_{\text{ERS}}$
is $4.73\%$ estimated using 500 samples.

\section*{Acknowledgments}

The authors are grateful to the Institute of Statistical Mathematics
and the Graduate Research Institute for Policy Studies in Tokyo where
a part of this research was carried out. We also thank Alex Shestopaloff
and James Thornton for their comments.


\begin{thebibliography}{99}
\bibitem{Adorisio2018}Adorisio, M., Pezzotta A., de Mulatier, C.,
Micheletti, C. and Celani, A. Exact and efficient sampling of conditioned
walks. \emph{Journal of Statistical Physics}, vol. 170, no. 1, pp.
79\textendash 100, 2018.

\bibitem{AndrieuDoucetHolenstein2010}C. Andrieu, A. Doucet and R.
Holenstein. Particle Markov chain Monte Carlo methods (with discussion).
\emph{Journal of the Royal Statistical Society }B, vol. 72, no. 3,
pp. 269\textendash 342, 2010.

\bibitem{CarterKohn1994}C.K. Carter and R. Kohn. On Gibbs sampling
for state space models. \emph{Biometrika}, vol. 81, no. 3, pp. 541\textendash 553,
1994.

\bibitem{CowellDawidLauritzenSpiegelhalter2006}R.G. Cowell, P. Dawid,
S.L. Lauritzen and D.J. Spiegelhalter.\emph{ Probabilistic Networks
and Expert Systems}. Springer-Verlag, Berlin-Heidelberg-New York,
1999.

\bibitem{ChopinDoucetRubenthaler2019}N. Chopin, A. Doucet and S.
Rubenthaler. Exact sampling using branching particle systems. Preprint,
2019.

\bibitem{DelMoral2004}P. Del Moral. \emph{Feynman-Kac Formulae: Genealogical
and Interacting Particle Systems with Applications}. Springer-Verlag:\ New
York, 2004.

\bibitem{DelMoralDoucet2004}P. Del Moral and A. Doucet. Particle
motions in absorbing medium with hard and soft obstacles. \emph{Stochastic
Analysis and Applications}, vol. 22, no. 5, pp. 1175\textendash 1207,
2004.

\bibitem{Devroye1986}L. Devroye. \emph{Non-Uniform Random Variate
Generation}. Springer-Verlag, 1986.

\bibitem{Douc2013}R. Douc, E. Moulines and D.S. Stoffer. \emph{Nonlinear
Time Series: Theory, Methods and Applications with R Examples}. CRC
Press, 2013.

\bibitem{Fearnhead2008}P. Fearnhead, O. Papaspiliopoulos and G.O.
Roberts. Particle filters for partially observed diffusions. \emph{Journal
of the Royal Statistical Society }B, vol. 70, no. 4, pp. 755\textendash 777,
2008.

\bibitem{Fill1998}J.A. Fill. An interruptible algorithm for perfect
sampling via Markov chains. \emph{The Annals of Applied Probability},
vol. 8, no. 1, pp. 131\textendash 162, 1998.

\bibitem{FinkeDoucetJohansen2016}A. Finke, A. Doucet and A.M. Johansen.
On embedded hidden Markov models and particle Markov chain Monte Carlo
Methods. arXiv preprint arXiv:1610.08962, 2016.

\bibitem{GodsillDoucetWest2004}S. J. Godsill, A. Doucet and M. West.
Monte Carlo smoothing for nonlinear time series. \emph{Journal of
the American Statistical Association}, vol. 99, no. 465, pp. 156\textendash 168,
2004.

\bibitem{Kendall2005}W.S. Kendall. Notes on perfect simulation. In\emph{
Markov chain Monte Carlo: Innovations and Applications} (eds. W.S.
Kendall, F. Liang and J.S. Wang), Lecture Notes Series, Institute
of Mathematical Science, National University of Singapore, pp. 93\textendash 146,
2005.

\bibitem{LinZhangChen2005}M.T. Lin, J.L. Zhang, Q. Cheng and R. Chen.
Independent particle filters. \emph{Journal of the American Statistical
Association}, vol. 100, no. 472, pp. 1412\textendash 1421, 2005.

\bibitem{Murdoch1998}D.J. Murdoch and P.J. Green. Exact sampling
from a continuous state space. \emph{Scandinavian Journal of Statistics},
vol. 25, no. 3, pp. 483\textendash 502, 1998.

\bibitem{Neal2003}R.M. Neal. Markov chain sampling for non-linear
state space models using embedded hidden Markov models. arXiv preprint
math/0305039, 2003.

\bibitem{Neal2004}R.M. Neal, M.J. Beal and S.T. Roweis. Inferring
state sequences for non-linear systems with embedded hidden Markov
models. \emph{Advances in Neural Information Processing Systems},
pp. 401\textendash 408, 2004.

\bibitem{PittShephard1999}M.K. Pitt and N. Shephard. Filtering via
simulation: Auxiliary particle filters. \emph{Journal of the American
Statistical Association}, vol. 94, no. 446, pp. 590-{}-599, 1999.

\bibitem{Poyiadjis2011}G. Poyiadjis, A. Doucet and S.S. Singh. Particle
approximations of the score and observed information matrix in state
space models with application to parameter estimation. \emph{Biometrika},
vol. 98, no. 1, pp. 65\textendash 80, 2011.

\bibitem{ProppWilson1996}J.G. Propp and D.B. Wilson. Exact sampling
with coupled Markov chains and applications to statistical mechanics.
\emph{Random Structures \& Algorithms}, vol. 9, no. 1\textendash 2,
pp. 223\textendash 252, 1996.

\bibitem{Rubin1988}D.B. Rubin. Using the SIR algorithm to simulate
posterior distributions. In \emph{Bayesian Statistics} 3, eds J.M.
Bernardo, M.H. DeGroot, D.V. Lindley and A.F.M. Smith, Oxford University
Press, pp. 395\textendash 402.

\bibitem{Shestopaloff2016}A.Y. Shestopaloff. MCMC Methods for Non-linear
State-Space Models. PhD thesis, Department of Statistics, University
of Toronto, 2016.

\bibitem{Smith1992}A.F.M. Smith and A.E. Gelfand. Bayesian statistics
without tears: a sampling-resampling perspective. \emph{The American
Statistician}, vol. 46, no. 2, pp. 84\textendash 88, 1992.

\bibitem{Wilkinson2002}D.J. Wilkinson and S.K.H. Yeung. Conditional
simulation from highly structured Gaussian systems, with application
to blocking-MCMC for the Bayesian analysis of very large linear models.
\emph{Statistics and Computing}, vol. 12, no. 3, pp. 287\textendash 300,
2002.
\end{thebibliography}
\end{document}